\documentclass[a4paper,12pt]{article}
\usepackage[utf8]{inputenc}
\usepackage{amssymb,amsthm,amsmath}
\usepackage{color}
\usepackage{graphicx}
\usepackage{graphics}
\usepackage{ifpdf}
\usepackage{epstopdf}

\usepackage[sort]{natbib}
\bibpunct{(}{)}{;}{a}{,}{,}

\newtheorem{theorem}{Theorem}[section]
\newtheorem{defi}[theorem]{Definition}
\newtheorem{lemma}[theorem]{Lemma}

\newtheorem{cor}[theorem]{Corollary}

\newtheorem{assum}[theorem]{Assumption}
\usepackage{authblk}

\begin{document}
\bibliographystyle{apalike}

\title{Market models with optimal arbitrage \footnote{We are grateful to Johannes Ruf for helpful comments on an earlier version of the paper. }}
\author[1,2]{Huy N. Chau}
\author[2]{Peter Tankov}
\affil[1]{University of Padova}
\affil[2]{LPMA, Université Paris-Diderot}
\date{}

\maketitle

\begin{abstract}
We construct and study market models
admitting optimal arbitrage. We say that a model admits optimal arbitrage if it is possible, in a zero-interest rate setting, starting with an initial wealth of $1$ and using only positive portfolios, to superreplicate a constant $c>1$. The optimal arbitrage strategy is the
strategy for which this constant has the highest possible value. Our definition of optimal arbitrage is similar to the one in  \cite{fernholz_optimal_2010}, where optimal relative arbitrage with respect to the market portfolio is studied. In this work we present a systematic method to construct market models where the
optimal arbitrage strategy exists and is known
explicitly. We then develop several new examples of market models with
arbitrage, which are based on economic agents' views concerning the
impossibility of certain events rather than ad hoc constructions. We also explore the concept of
\emph{fragility} of arbitrage introduced in \cite{guasoni_fragility_2012}, and provide
new examples of arbitrage models which are not fragile in this sense. 
\end{abstract}


\section{Introduction}
A key concept in mathematical finance is that of absence of arbitrage. Informally speaking, an arbitrage opportunity is the possibility of making money out of nothing without taking any risk. Clearly, such strategies should be excluded in order to ensure market efficiency. The mathematical formulation of the no-arbitrage theory has a long history from the first discrete-time result of \cite{harrison1979martingales} and to the characterizations of  \cite{delbaen_general_1994} and \cite{delbaen_fundamental_1998} in general semimartingale models. In particular, in \cite{delbaen_general_1994} it is shown that the No Free Lunch With Vanishing Risk (NFLVR) condition is equivalent to the existence of an Equivalent Local Martingale Measure (ELMM), i.e., a new probability measure under which the discounted asset price process is a local martingale.

The NFLVR provides a sound theoretical framework to solve problems of pricing, hedging or portfolio optimization. However, for some applications, requiring total absence of free lunches turns out to be too restrictive. Indeed, it seems reasonable to assume that limited arbitrage opportunities exist in financial markets given that whole desks of ``arbitrageurs'' are working full-time in investment banks to exploit them.   This is one of the reasons why market models with arbitrage opportunities have appeared in recent literature, starting with the three-dimensional Bessel process model of \cite{delbaen_arbitrage_1995}. Without relying on the concept of equivalent martingale measure, \cite{platen_benchmark_2006}, see also \cite{platen_various_2006}, developed the Benchmark Approach, a new asset pricing theory under physical measure. In the context of Stochastic Portfolio Theory \citep{karatzas_stochastic_2009}, the NFLVR condition is not imposed and arbitrage opportunities arise in relative sense.  These works suggest that NFLVR condition can be replaced by another weaker notion while preserving the solvability of the economics problems mentioned above,  and several new concepts have indeed been proposed, see \cite{fontana_weak_2013} for a review.
 
If one is interested in utility maximization, it has been shown \citep{karatzas_numeraire_2007} that the minimal no free lunch type condition making this problem well posed is the No Unbounded Profit with Bounded Risk (NUPBR) condition. This condition has also been referred to as BK in \cite{kabanov1997ftap} and it is also equivalent to the No Asymptotic Arbitrage of the 1st kind (NAA1) of \cite{kabanov1994large}. 
It is known that the NFLVR is equivalent to NUPBR plus the
classical no arbitrage assumption (see Corollary 3.4 and Corollary 3.8 of \cite{delbaen_general_1994} or Proposition 4.2 of \cite{karatzas_numeraire_2007}). This means that markets satisfying only NUPBR may admit arbitrage opportunities. 

To benefit from a potential arbitrage, one needs to characterize explicitly the arbitrage strategy, and also to devise a method to compare different strategies, so as to exploit the arbitrage opportunity in the most efficient way. 
An important step in this direction was made in \cite{fernholz_optimal_2010}. In this paper, the authors introduce the notion of optimal relative arbitrage with respect to the market portfolio and characterize the optimal relative arbitrage in continuous Markovian market models in terms of the smallest positive solution to a parabolic partial differential inequality. This work is extended in \cite{fernholz_optimal_2011} to market models with uncertainty regarding the relative risk and covariance structure of the assets. 
Precise conditions for the existence of both relative arbitrage and strong relative arbitrage opportunities are given in \cite{ruf_optimal_2011}. 
The optimal relative arbitrage turns out to be related to the minimal cost needed to superhedge the market portfolio in an almost sure way. 
The problem of hedging in markets with arbitrage opportunities is studied in detail in \cite{ruf_hedging_2013}. That paper shows in particular that delta hedging is still the optimal hedging strategy in continuous Markovian markets which admit no equivalent local martingale measure but only a square-integrable market price of risk. 

With the exception of the Stochastic Portfolio Theory of Fernholz and Karatzas and the Benchmark Approach of Platen and Heath, the existing literature on markets without arbitrage opportunities remains theoretical, and its findings are rarely used by practitiones to develop actual trading strategies based on arbitrage opportunities. This is related to the fact that examples of market with arbitrage found in the literature are typically ad hoc and are not flexible enough allow calibration to the actual market data. Moreover, it has been shown in \cite{guasoni_fragility_2012} that most of the existing examples of arbitrage in diffusion markets are \emph{fragile}, meaning that they disappear when transaction costs are introduced into the model, however small their value, or when prices are recorded with a small observation error. 

The goal of this study is to propose a new methodology for building models admitting optimal arbitrage, with an explicit characterization of the optimal arbitrage strategy. To do so, we start with a probability measure $\mathbb Q$ under which the NFLVR condition holds. We then construct a new probability measure $\mathbb P$, not equivalent to $\mathbb Q$, under which NFLVR no longer holds but NUBPR is still satisfied. This procedure is not new and goes back to the construction of the Bessel process by \cite{delbaen_arbitrage_1995}. However, we extend it in two directions.

Firstly, from the theoretical point of view, we provide a characterization of the superhedging price of a claim under $\mathbb P$ in terms of the superhedging price of a related claim under $\mathbb Q$. This allows us to characterize the optimal arbitrage profit under $\mathbb P$ in terms of the superhedging price under $\mathbb Q$, which is much easier to compute using the equivalent local martingale measures. 

Secondly, from the economic point of view, we provide an economic intuition for the new arbitrage model as a model implementing the view of the economic agent concerning the impossibility of certain market events. In other words, if an economic agent considers that a certain event (such as a sovereign default) is impossible, but it is actually priced in the market, our method can be used to construct a new model incorporating this arbitrage opportunity, and to compute the associated optimal arbitrage strategy. 

We then combine these two ideas to develop several new classes of examples of models with optimal arbitrage, allowing for a clear economic interpretation, with a special focus on incomplete markets. We also discuss the issue of robustness of these arbitrages to small transaction costs / small observation errors and show that some of our examples are not fragile in the sense of \cite{guasoni_fragility_2012}. 

The paper is organized as follows. In Section \ref{section:setting}, we describe the market setting and state the main assumptions. In Section \ref{section:optimal_arbitrage}, optimal arbitrage profit is introduced and related to a superhedging problem. In Section \ref{section:main_result}, we use an absolutely continuous measure change to build markets with optimal arbitrage. Finally, several new examples built using this construction are gathered in Section \ref{section:examples}. 

\section{General setting} \label{section:setting}
For the theory of stochastic process and stochastic integration, we refer to \cite{jacod_limit_2002} and \cite{protter_stochastic_2003}. 

Let $(\Omega, \mathcal F, \mathbb{F}, \mathbb P)$ be a given fitered probability space, where the filtration $\mathbb{F} = (\mathcal F_t)_{t\geq 0}$ is assumed to satisfy the usual conditions of right-continuity and augmentation by the $\mathbb P$-null sets. For any adapted RCLL process $S$, we denote by $S_-$ its predictable left-continuous version and by $\Delta S:=S - S_-$ its jump process. For a $d$-dimensional semimartingale $S$ and a predictable process $H$, we denote by $ H \cdot S$ the vector stochastic integral of $H$ with respect to $S$. We fix a finite planning horizon $T<\infty$ (a stopping time) and assume that after $T$ all price processes are constant and equal to their values at $T$. 

On the stochastic basic $(\Omega, \mathcal F, \mathbb{F}, \mathbb P)$, we consider a financial market with an $\mathbb R^d$-valued nonnegative semimartingale process $S = (S^1,...,S^d)$ whose components model the prices of $d$ risky assets. The riskless asset is denoted by $S^0$ and we assume that $S^0 \equiv 1$, that is, all price processes are already discounted. We suppose that the financial market is frictionless, meaning that there are no trading restrictions, transaction costs, or other market imperfections. 

Let $L(S)$ be the set of all $\mathbb{R}^d$-valued $S$-integrable predictable processes. It is the most reasonable class of strategies that investors can choose, but another constraint, which is described below, is needed in order to rule out doubling strategies. 
\begin{defi}
Let $x \in \mathbb{R_+}$. An element $H \in L(S)$ is said to be an $x$-admissible strategy if $H_0
= 0$ and $(H \cdot S)_t \ge - x$ for all $t \in [0, T ]$ $\mathbb P$-a.s. An element $H \in L(S)$ is said to be an admissible strategy if it is an $x$-admissible strategy for some $x \in \mathbb{R_+}$.
\end{defi}
For $x \in \mathbb{R_+}$, we denote by $\mathcal{A}_x$ the set of all $x$-admissible strategies and by $\mathcal{A}$ the set of all
admissible strategies. As usual, $H_t$ is assumed to represent the number of units of risky asset that we hold at time $t$. For $(x,H) \in \mathbb{R_+} \times \mathcal{A}$,
we define the portfolio value process $V_t^{x,H} = x + (H \cdot S)_t$. This is equivalent to require that portfolios are only generated by self-financing admissible strategies.

Given the semimartingale $S$, we denote by $\mathcal{K}_x$ the set of all outcomes that one can realize by $x$-admissible strategies starting with zero initial cost:
$$ \mathcal{K}_x = \left\lbrace (H \cdot S)_T| H  \text{ is $x$-admissible} \right\rbrace$$
and by $\mathcal{X}_x$ the set of outcomes of strategies with initial cost $x$:
$$\mathcal{X}_x = \left\lbrace x + ( H \cdot S)_T| H  \text{ is $x$-admissible} \right\rbrace.$$
Remark that all elements in $\mathcal{X}_x$ are nonnegative. The unions of all $\mathcal{K}_x$ and all $\mathcal{X}_x$ are denoted by $\mathcal{K}$
and $\mathcal{X}$, respectively. All bounded claims which can be superreplicated by admissible strategies are contained in 
$$ \mathcal{C} = \left( \mathcal{K} - L^0_+ \right) \cap L^{\infty}.$$
Now, we recall some no-free-lunch conditions, which are studied in the works of \cite{delbaen_general_1994},  \cite{karatzas_numeraire_2007} and \cite{kardaras_market_2012}.
\begin{defi}
\begin{itemize}
\item We say that the market satisfies the No Arbitrage (NA) condition with respect to general  admissible integrands if  
\[
\mathcal{C} \cap L_+^{\infty} = \left\lbrace 0 \right\rbrace .
\]
\item We say that the market satisfies the No Free Lunch with Vanishing Risk (NFLVR) property, with respect to general admissible integrands, if 
$$ \overline{\mathcal{C}} \cap L^{\infty}_+ = \left\lbrace 0 \right\rbrace  ,$$
where the bar denotes the closure in the supnorm topology of $L^{\infty}$.
\item There is No Unbounded Profit With Bounded Risk (NUPBR) if the set $ \mathcal{K}_1$ is bounded in $L^0$, that is, if
$$ \lim_{c \to \infty} \sup_{H \in \mathcal{A}_1} \mathbb P\left[ V^{0, H} > c \right] =0$$ 
holds.
\item The market admits Immediate Arbitrage (IA) if there exists a stopping time $\tau$ such that $\mathbb P\left[ \tau < T \right] >0 $ and a strategy $H = H 1_{]\tau, T] }$ which realizes $\left( H\cdot S \right)_t > 0 $ $\mathbb P$-a.s for all $t \in (\tau, T]$. We say that condition No Immediate Arbitrage (NIA) holds if there exists no immediate arbitrage in the market.
\end{itemize}
\end{defi}
The economic interpretation of no arbitrage type conditions above can be described as follows. Classical arbitrage means that one can make something out of nothing without risk. If there is a FLVR, starting with zero capital, one can find a sequence of wealth processes such that the terminal wealths converge to a nonnegative random variable which is not identical to zero and the risk of the trading strategies becomes arbitrarily small. If an UPBR exists, one can find a sequence of wealth processes with bounded (or indeed arbitrarily small) risk whose terminal wealths are unbounded with a fixed probability.

In this paper, we are interested in financial markets satisfying the following assumption.
\begin{assum}\label{assume:NUPBR} S is locally bounded, the market satisfies
NUPBR but the condition NFLVR fails under the physical measure $\mathbb P$.
\end{assum}

Under the local boundedness assumption, by the Fundamental Theorem of Asset Pricing, the NFLVR condition is equivalent to the existence of a ELMM (see Corollary 1.2 in \cite{delbaen_general_1994}), but for general semimartingales the ELMM must be replaced with an equivalent sigma-martingale measure. So, the limitation to locally bounded processes $S$ allows us to work with local martingales instead of sigma martingales.

When the NFLVR condition fails but the NUBPR condition holds, the ELMM is replaced with a weaker notion of ``deflator''. 

\begin{defi}
An equivalent local martingale deflator (ELMD) is a nonnegative process $Z$ with $Z_0=1$ and $Z_T > 0$ such that $ZV$ is a local martingale for all $V \in \mathcal{X}$.
\end{defi}
In particular, an ELMD is a nonnegative local martingale. Fatou's Lemma implies  that it is also a supermartingale and its expectation is less or equal to one. Hence, a local martingale density is an ELMD with expectation one. In general, we cannot use an ELMD to define a new probability measure, since the new measure loses mass. It is worth to remark that the situation when the ELMD is a strict local martingale is very different from a market with bubble. Indeed, an asset price is said to be a bubble if it is a strict local martingale under the risk-neutral measure, see \cite{heston_options_2006} , \cite{cox_local_2005}, \cite{jarrow_asset_2007}, \cite{jarrow_asset_2010}, which means that no arbitrage opportunity exists. 

The following result has recently been proven in \cite{kardaras_market_2012} in the one dimensional case. An alternative proof in the multidimensional case has been given in \cite{takaoka_condition_2010} by a suitable change of numeraire argument in order to apply the classical results of \cite{delbaen_general_1994}, and in \cite{song_alternative_2013} by only using the properties of the local characteristics of the asset process.
\begin{theorem}
The NUPBR condition is equivalent to the existence of at least one ELMD.
\end{theorem} 

\paragraph{Fragile and robust arbitrages} Since real markets have frictions, an arbitrage which disappears in the presence of small transaction costs, or small observation errors, cannot be exploited in practice. This property, known as fragility of arbitrage \citep{guasoni_fragility_2012} is described in the following two definitions.
{ 
\begin{defi}
For $\varepsilon >0$, two strictly positive processes $S, \tilde{S}$ are $\varepsilon$-close if
$$ \frac{1}{1 + \varepsilon} \leq \frac{\tilde{S}_t}{S_t} \leq 1 + \varepsilon \qquad a.s. \text{ for all } t \in [0,T].$$
\end{defi}

\begin{defi}[Fragility/Robustness]
We say that arbitrage in the $(\mathbb{P},S)$-market is fragile if for every $\varepsilon >0$ there exists a process $\tilde{S}$, which is $\varepsilon$-close to $S$, such that the $(\mathbb{P}, \tilde{S})$-market satisfies NFLVR. If the arbitrage is not fragile we say that it is robust. 
\end{defi}
\cite{guasoni_fragility_2012} show that in diffusion settings, if the coefficients of the log-price process are locally bounded, arbitrages are fragile. For instance, when we introduce small frictions in the Bessel example, the arbitrage disappears.

\cite{bender_simple_2012} defines a \emph{simple obvious arbitrage} as a buy and hold strategy, which guarantees to the investor a profit of at least $\varepsilon>0$ if the investor trades at all. It is clear that a simple obvious arbitrage is always robust. 
}
\section{Optimal arbitrage} \label{section:optimal_arbitrage}
It is well known that NFLVR holds if and only if both NUPBR and NA hold, see Corollary 3.4 and 3.8 of \cite{delbaen_general_1994} or Proposition 4.2 of \cite{karatzas_numeraire_2007}. Moreover, Lemma 3.1 of \cite{delbaen_existence_1995} shows that if NA fails then either the market admits an immediate arbitrage or an arbitrage that is created by a strategy in $\mathcal{A}_1$. Furthermore, if there exists an immediate arbitrage, the associated strategy is in $\mathcal{A}_0$ and can be freely scaled to produce an unbounded arbitrage. But this situation is not allowed in our market due to  Assumption \ref{assume:NUPBR}. Therefore, it is only possible to exploit arbitrages by using strategies in the set $\mathcal{A}_x$ where $x>0$.
This reasoning is formalized by the following lemma.
\begin{lemma} \label{lemma:NUPBR_NIA}
NUPBR implies NIA.
\end{lemma}
\begin{proof}
We will prove that immediate arbitrage implies unbounded profit with bounded risk. Assume $S$ admits immediate arbitrage at stopping time $\tau$ and $\mathbb P(\tau < T) > 0$. There exits a strategy $H$ such that $H = H1_{\left( \tau, T \right] }$ and $(H\cdot S)_t > 0$ for $t \in \left( \tau, T \right]$. We observed that $H^n:= nH \in \mathcal{A}_0$ and the set $\left\lbrace V^{0,H^n}_T \right\rbrace_{n \in \mathbb{N}}$ is not bounded in probability. This means that the market admits unbounded profit with bounded risk.
\end{proof}
For these reasons, arbitrages in our market are limited and the question of optimal arbitrage profit arises naturally. 
\begin{defi} \label{defi:optimal_arbitrage}
For a fixed time horizon $T$, we define 
$$U(T) : = \sup \left\{ {c > 0:\exists H  \in \mathcal{A}_1,V_T^{1,H}   \ge c, \mathbb P-a.s} \right\}.$$
If $U(T) > 1$, we call $U(T)$ optimal arbitrage profit.
\end{defi}
The quantity $U(T)$ is the maximum deterministic amount that one can realize at time T starting from unit initial capital. Obviously, this value is bounded from below by 1. This definition goes back to the paper of \cite{fernholz_optimal_2010}. In diffusion setting, these authors characterize the following value
$$ \sup \left\lbrace c > 0: \exists H \in \mathcal{A}_1, V_T^{1,H} \ge c \sum S^i_T, \mathbb P-a.s. \right\rbrace,$$
which is the highest return that one can achieve relative to the market capitalization.
\subsection{Optimal arbitrage and superhedging price}
\begin{defi}\label{defi:superhedge_price}
Given a claim $f \ge 0$, we define
$$SP_+(f) := \inf \left\{ x \ge 0: \exists H \in \mathcal{A}_x ,V_T^{x,H}  \ge f, \mathbb P-a.s \right\},$$
that is the minimal amount starting with which one can superhedge $f$ by a nonnegative wealth process.
\end{defi}

A reason for limiting oneself to nonnegative wealth processes is discussed in \cite{ruf_optimal_2011}. If $Z$ is a local martingale deflator and a wealth process $V$ is only restricted to stay above some constant $x < 0$, then $ZV$ may no longer be a supermartingale. Let us compare the definition of $SP_+$-price with the superhedging price in the literature. The superhedging price of a given claim $f$ is commonly defined by
$$SP(f) = \inf \left\{ x \ge 0: \exists H \in \mathcal{A} ,V_T^{x,H}  \ge f, \mathbb P-a.s \right\}.$$
In other words, we allow wealth processes (maybe negative) that are uniformly bounded from below. Note that in markets that satisfy NA, $SP_+(f) = SP(f)$. Indeed, if NA holds, for every admissible integrand $H$ we have $ \| (H \cdot S)_t^{-} \|_{\infty} \le \| (H \cdot S)_T^{-} \|_{\infty}$, see Proposition 3.5 in \cite{delbaen_general_1994}. If $ x + (H \cdot S)_T \ge f$ then $(H \cdot S)_T \ge f - x \ge -x$ so that $\| (H \cdot S)_T^{-} \|_{\infty} \le x$. This implies that $ \| (H \cdot S)_t^{-} \|_{\infty} \le x$ or $(H \cdot S)_t \ge -x,$ for all $t \in [0,T]$. 
In our market model with arbitrage, $SP(f) \le SP_+(f)$. The difference between the two superhedging prices is discussed in \cite{khasanov_upper_2013}.

The following lemma is simple but useful to our problem.
\begin{lemma} $U(T) = 1/SP_+(1)$.
\end{lemma}
\begin{proof}
$(\le)$ Take any $c>0$ such that there exists a strategy $H \in \mathcal{A}_1$ which satisfies
\begin{itemize}
\item $V_T^{1,H } = 1 + (H\cdot S)_T  \ge c, \mathbb P-a.s.$
\item $(H \cdot S)_t \ge -1.$ for all $0 \le t \le T$.
\end{itemize}
Then a simple scaling argument gives us a strategy to hedge 1
\begin{itemize}
\item (superheging) $1/c + 1/c(H\cdot S)_T  \ge 1, \mathbb P-a.s.$
\item (admissibility) $1/c(H\cdot S)_t \ge -1/c$ for all  $0 \le t \le T$.
\end{itemize}
By Definition \ref{defi:superhedge_price}, one can superhedge 1 at cost $1/c$. Therefore, we get an upper bound for optimal arbitrage profit
\[
U(T) \le \frac{1}{SP_+(1)}.
\]
$(\ge)$ The converse inequality can be proved by the same argument.
\end{proof}
The above lemma has two consequences. First, the knowledge of $SP_+(1)$ is enough to find optimal arbitrage profit. Second, one should find the strategy to superhedge 1 in order to realize optimal arbitrage.

Obviously, $SP_+(1) \le 1$. If $SP_+(1) < 1$, there is optimal arbitrage. If $SP_+(1) = 1$, optimal arbitrage does not exist, but arbitrages may still exist. In Example 8 of \cite{ruf_optimal_2011}, the cheapest price to hold $1$ is $1$, but we can achieve a terminal wealth larger than $1$ with positive probability.

\section{Constructing market models with optimal arbitrage}
\label{section:main_result}
In this section we present two constructions of market models with optimal arbitrage. They both work by starting with a probability measure $\mathbb Q$ under which the price process satisfies NFLVR and making a non-equivalent measure change to construct a new measure $\mathbb P$ allowing for arbitrage. Arbitrage opportunities constructed with an absolutely continuous measure change have been studied in earlier works. The first example of this kind of technique is the Bessel model, which is given in  \cite{delbaen_arbitrage_1995}. This technique is generalized in  \cite{osterrieder_arbitrage_2006} and \cite{ruf_systematic_2013}. However, we push this idea further by characterizing the superhedging price under $\mathbb P$ in terms of the superhedging price under $\mathbb Q$, which enables us to describe optimal arbitrages. 
\subsection{A construction based on a nonegative martingale} \label{section:construction_martingale}
Let $\mathbb Q$ be a probability measure on the filtered measure space $\left( {\Omega ,\mathcal F , \left( \mathcal F_t \right)_{t\geq 0}}\right)$ described in the beginning of Section \ref{section:setting}, and assume that under $\mathbb Q$, the following are true:
\begin{itemize}
\item The risky asset process $S$ is a locally bounded semimartingale which satisfies NFLVR. 
\item There exists a nonnegative RCLL martingale $M$ with $M_0 = 1$ and 
\begin{equation} \label{eq:condition_on_M}
{\mathbb Q}[\{\tau \le T\} \cap \{ M_{\tau -} > 0 \}] = 0.
\end{equation}
where $\tau = \inf \{t \ge 0: M_t = 0\}$ with the convention that $\inf \emptyset = + \infty.$
\end{itemize}
Since $M$ is right-continuous, this means that $M$ may only hit zero continuously on $[0,T]$.
Using $M$ as a Radon-Nykodym derivative, we define a new probability measure via 
$$
\frac{d\mathbb P}{d{\mathbb Q}}\Big|_{\mathcal F_t} = M_{\tau \wedge t}.
$$ 
Then $\mathbb P$ is only absolutely continuous (but not equivalent) with respect to ${\mathbb Q}$. In fact, $M$ can reach zero under ${\mathbb Q}$ but it is always positive under $\mathbb P$, because $\mathbb P[\tau \le T] = {\mathbb E}^{\mathbb Q}[M_{\tau \wedge T}1_{\tau \le T}] = 0$. 

\begin{theorem} \label{theorem:optimal_arbitrage_constrution_martingale}
Under the above assumptions, the $(\mathbb P,S)$-market satisfies NUPBR, and for any $\mathcal F_T$-measurable claim $f \ge 0$, we have  $$SP^\mathbb P_+(f) = SP^{\mathbb Q}_+(f1_{M_T > 0}).$$ 
\end{theorem}
\begin{cor}
Under the assumptions of the theorem let 
$$
\sup_{\bar {\mathbb Q} \in ELMM({\mathbb Q},S)} {\mathbb E}^{\bar {\mathbb Q}} [1_{M_T>0}]  < 1.
$$
Then the $(\mathbb P,S)$-market admits optimal arbitrage and the optimal arbitrage strategy the superhedging strategy of the claim $1_{M_T > 0}$ in the $(\mathbb{Q},S)$-market.
\end{cor}
\begin{proof}
By the the standard super-replication theorem under absence of arbitrage (Theorem 9 of \cite{delbaen_no-arbitrage_1995}),
$$SP^{\mathbb Q}_+(1_{M_T > 0}) = \sup_{\bar {\mathbb Q} \in ELMM({\mathbb Q},S)} {\mathbb E}^{\bar {\mathbb Q}} [1_{M_T}>0].
$$
\end{proof}
\begin{proof}[Proof of Theorem]Let $\bar {\mathbb Q}$ be a local martingale measure equivalent to $\mathbb Q$, and denote by $\bar Z$ its density with respect to $\mathbb Q$. 
\\
Step 1: we prove that the $(\mathbb P,S)$-market satisfies NUPBR by showing that $\bar Z/M$ is an ELMD. 

We define 
$$\tau_n = \inf\{ t \geq 0: M_t < \frac{1}{n}\}$$
with the convention $\inf \emptyset = + \infty$. Since, by condition (\ref{eq:condition_on_M}), $M$ does not jump to zero, we have that $M_{t\wedge \tau_n}>0$ $\forall t\geq 0$ $\mathbb Q$-a.s.

We remark that ${\mathbb Q} \ll {\mathbb P}$ on $\mathcal F_{t \wedge \tau_n}$. Indeed, take any $A \in \mathcal F_{t \wedge \tau_n}$ such that ${\mathbb P}(A) = 0$, we compute 
\begin{equation*}
{\mathbb Q}[A] = {\mathbb E}^{\mathbb Q}\left[ 1_A \frac{M_{t \wedge \tau_n}}{M_{t \wedge \tau_n}} \right] = {\mathbb E}^{\mathbb P}\left[ 1_A \frac{1}{M_{t \wedge \tau_n}} \right]  = 0.
\end{equation*}
This means ${\mathbb P}$ is equivalent to ${\mathbb Q}$ on $\mathcal F_{t \wedge \tau_n}$. 

By Corollary 3.10, page 168 of \cite{jacod_limit_2002}, to prove that a process $N$ is a $\mathbb P$-local martingale with localizing sequence $(\tau_n)$, we need to prove that $(NM)^{\tau_n}$ is a $\mathbb Q$-local martingale for every $n\geq 1$. Then,
\begin{itemize}
\item $\bar Z/M$ is a $\mathbb P$-local martingale since $\bar Z$ is a $\mathbb{Q}$-local martingale. 
\item $\bar Z S/M$ is a $\mathbb P$-local martingale since $\bar ZS$ is a $\mathbb Q$-local martingale.
\item $\bar Z V/M$ is a $\mathbb P$-local martingale for each $\mathbb P$-admissible $V$. Since $\mathbb{P}$ and $\mathbb{Q}$ are equivalent on $\mathcal F_{t \wedge \tau_n}$, we obtain $V^{\tau_n}$ is a $\mathbb{Q}$-admissible wealth process. Thus, for each $n$, we have $\bar Z V^{\tau_n}$ is a $\mathbb{Q}$-local martingale, so is $(\bar Z V)^{\tau_n}$.
\end{itemize}

\noindent Step 2: we prove the equality $SP_+^\mathbb P(f) = SP_+^{\mathbb Q}(f1_{M_T >0}).$\\ 
($\le$) Take any $x >0$ such that there exists a strategy $H \in \mathcal{A}^{\mathbb Q}_x$ which satisfies $ V_T = x + (H \cdot S)_T \ge f1_{M_T >0}, {\mathbb Q}-a.s.$ Since $\mathbb P \ll {\mathbb Q}$, Theorem 25, page 170 of \cite{protter_stochastic_2003} shows that $H \in L(S)$ under $\mathbb P$ as well and $H_{\mathbb Q} \cdot S = H_{\mathbb P} \cdot S, {\mathbb P}-a.s.$. We also see that $ x + (H \cdot S)_t \ge 0, {\mathbb P}-a.s$ and $ x + (H \cdot S)_T \ge f1_{M_T >0} = f, {\mathbb P}-a.s$. This means 
\begin{equation} \label{eq:SPlsSQ}
SP_+^{\mathbb P}(f) \le SP_+^{\mathbb Q}(f1_{M_T > 0}).
\end{equation}
($\ge$) For the converse inequality, take any $x >0$ such that there exists a strategy $H \in \mathcal{A}^{\mathbb P}_x$ and $ V^{\mathbb P}_T = x + (H \cdot S)_T \ge f, {\mathbb P}-a.s.$ We will show that $x \ge SP^{\mathbb Q}_+(f1_{M_T > 0 })$.\\
Define  $H^n:=H1_{ t \le \tau_n}$, then $H^n$ is $S$-integrable and $x$-admissible under ${\mathbb Q}$. From Step 1, we see that $\tau_n \wedge T \nearrow T$, ${\mathbb P}$-a.s. and therefore $V^n_T = x + (H^n \cdot S)_T \to V^{\mathbb P}_T$, ${\mathbb P}$-a.s. or 
$V^n_T 1_{M_T > 0} \to V^{\mathbb P}_T1_{M_T > 0} \ge f 1_{M_T > 0}$, ${\mathbb Q}$-a.s. The following convergence holds
$$V^n_T - V^n_T 1_{M_T = 0} = V^n_T 1_{M_T >0} \to  V^{\mathbb P}_T1_{M_T >0} \ge f 1_{M_T >0}, {\mathbb Q}-a.s.$$
The sequence $V^n_T - x - V^n_T 1_{M_T = 0} = (H^n \cdot S)_T - V^n_T 1_{M_T >0}$ is in the set $\mathcal{K} - L^0_+$ (under ${\mathbb Q}$) and uniformly bounded from below by $-x$. Because the $({\mathbb Q},S)$-market satisfies NFLVR condition, the set $\mathcal{K} - L^0_+$ is Fatou-closed (see Remark after Corollary 1.2 of \cite{delbaen_general_1994}) and we obtain $V^{\mathbb P}_T1_{M_T > 0} - x \in \mathcal{K} - L^0_+$ or $x \ge SP^{\mathbb Q}_+(f1_{M_T > 0})$. In other words, 
\begin{equation} \label{eq:SPgtSQ}
SP^{\mathbb P}_+(f) \ge SP^Q_+(f1_{M_T > 0}).
\end{equation}
From (\ref{eq:SPlsSQ}) and (\ref{eq:SPgtSQ}), the proof is complete.
\end{proof}

\subsection{A construction based on a predictable stopping time}\label{section:construction_predictable}
In this section, we apply the construction of Section \ref{section:construction_martingale} to a predictable stopping time. As before, we consider a measure $\mathbb Q$ on the space $(\Omega,\mathcal F, (\mathcal F)_{t\geq 0})$. Let $\sigma$ be a stopping time such that ${\mathbb Q}\left( \sigma > T \right) >0.$ We define a new probability measure, absolutely continuous with respect to ${\mathbb Q}$, by
\begin{align}\left. {\frac{d{\mathbb P}}{d{\mathbb Q}}} \right|_{\mathcal F_t}  =  \frac{{\mathbb Q}\left[ \sigma > T | \mathcal F_t \right] }{{\mathbb Q}\left[ \sigma > T \right] } := M_t.\label{defprobatau}
\end{align}
Under the new measure, ${\mathbb P}\left( \sigma > T \right) = {\mathbb E}^{\mathbb Q}\left( M_T 1_{\sigma > T} \right) = 1$.\\


This construction has the following economic interpretation. Consider an event (E), such as the default of a company or a sovereign state, whose occurence is characterized by a stopping time $\sigma$. Given a planning horizon $T$, we are interested in the occurence of this event (E) before the planning horizon. Suppose that the market agents have common anticipations of the probability of future scenarios, which correspond to the arbitrage-free probability measure $\mathbb Q$, and that under this probability, the event (E) has nonzero probabilities of occuring both before and after the planning horizon. Consider now an informed economic agent who believes that the event (E) will not happen before the planning horizon $T$. For instance, the agent may believe that the company or the state in question will be bailed out in case of potential default. Our informed agent may then want to construct an alternative model $\mathbb P$, in which the arbitrage opportunity due to mispricing may be exploited and the arbitrage strategy may be constructed in a rigorous manner.\footnote{The ``informed agent'' interpretation of our arbitrage construction hints at possible connections with the research on arbitrage opportunities arising from enlargement of the underlying filtration with additional information, see e.g. \cite{imkeller2001free,fontana2012arbitrages}. The detailed study of these connections is left for further research. } The following corollary provides a method for constructing such a model. 
\begin{cor} \label{cor:optimal_arbitrage_predictable}
Assume that the following conditions hold
\begin{itemize}
\item The risky asset process $S$ is a locally bounded semimartingale which satisfies NFLVR under $\mathbb Q$.
\item The filtration $\mathbb{F}$ is quasi-left continuous.
\item $\sigma$ is a predictable stopping time such that for any stopping time $\theta$, 
$${\mathbb E}^{\mathbb Q}\left[ 1_{\sigma > T}\right |\mathcal F_{\theta} ] > 0, {\mathbb Q}-a.s. \qquad \text{on } \left\lbrace \sigma > \theta \right\rbrace.$$ 
\end{itemize} 
Then the $({\mathbb P}, S)$-market satisfies NUPBR. Given a $\mathcal F_T$-measurable claim $f \ge 0$, we have
$$SP^{\mathbb P}_+(f) = SP^{\mathbb Q}_+(f1_{\sigma > T}).$$ 
In addition, if
$$SP^{\mathbb Q}_+(1_{\sigma > T}) = \sup_{\bar {\mathbb Q} \in ELMM({\mathbb Q},S)} {\mathbb E}^{\bar {\mathbb Q}} [1_{\sigma > T}]  < 1,$$
then the $({\mathbb P},S)$-market admits optimal arbitrage.
\end{cor}
\begin{proof}
This result will follow from Theorem \ref{theorem:optimal_arbitrage_constrution_martingale} after checking the condition (\ref{eq:condition_on_M}) on $M$. Let $\tau = \inf\{t>0: M_t = 0\}$. By construction, $M_\sigma = 0$ on $\{\sigma \le T\}$ and $M_t>0$ for $t<\sigma$. This means that
$$
\tau= \left\{
\begin{aligned}
&\sigma,\quad &&\sigma \leq T\\
&+\infty, && \text{otherwise.}
\end{aligned}\right.
$$
Since the filtration $\mathbb{F}$ is quasi left continuous and $\sigma$ is a predictable stopping time, $M$ does not jump at $\sigma$ (see \cite{protter_stochastic_2003}, page 190).  This means that
$$
\mathbb Q[\{\tau\leq T\}\cap \{M_{\tau-}>0\} ] = 0
$$
and condition \eqref{eq:condition_on_M} is satisfied.
\end{proof}
\section{Examples} \label{section:examples}
\subsection{A complete market example}
Let $W^{\mathbb Q}$ be a Brownian motion and let  $\mathbb{F}$ be its completed natural filtration. We assume that the price of a risky asset evolves as follows
$$
S_t = 1 + W^{\mathbb Q}_t
$$
and define a predictable stopping time by $\sigma = \inf\{ t>0: S_t \le 0\}$. Using the law of infimum of Brownian motion, we get
$${\mathbb Q}[\sigma > T] = {\mathbb Q}[(W^{\mathbb Q})^*_T > -1] = 1 - 2\mathcal{N} \left( -\frac{1}{\sqrt{T}} \right) > 0,$$
where $\mathcal N$ denotes the standard normal distribution function. 

Next, by Markov property we compute
\begin{align}
{\mathbb E}^{\mathbb Q}[1_{\sigma > T}|\mathcal F_t]  = {\mathbb Q}[(W^{\mathbb Q})^*_T > -1|\mathcal F_t] = \left\{ {\begin{array}{ll}
   {0} \qquad \text{on  } {\sigma \le t}   \\
   {1 - 2\mathcal{N} \left( -\frac{S_t}{\sqrt{T-t}} \right) > 0} \qquad \text{on }\sigma>t.\\
\end{array}} \right.\label{martrep}
\end{align}
Hence, on $\{\tau > t\}$, we obtain ${\mathbb E}^{\mathbb Q}[1_{\sigma > T}|\mathcal F_t] >0$. This means that the construction of Section \ref{section:construction_predictable} applies and we may define a new measure $\mathbb P$ via \eqref{defprobatau}. Since the $({\mathbb Q},S)$-market is complete and $ELMM({\mathbb Q},S) = \{{\mathbb Q}\}$, the superhedging price of the claim $1_{\sigma > T}$ is
$${\mathbb Q}[\sigma > T] =  1 - 2\mathcal{N} \left( -\frac{1}{\sqrt{T}} \right) < 1,$$
which means that the $\mathbb P$-market admits optimal arbitrage. 

Applying the Itô formula to \eqref{martrep}, we get the martingale representation:
\begin{align} \label{example:BM_dynamic_hedging} 
{\mathbb E}^{\mathbb Q}[1_{\sigma > T}|\mathcal F_t] = {\mathbb Q}[\sigma > T] + \sqrt{\frac{2}{\pi}} \int\limits_0^{\sigma\wedge t} {\frac{1}{\sqrt{T-s}} e^{-\frac{S_s^2}{2(T-s)}} dW^{\mathbb Q}_s}.
\end{align}
Therefore,
$$ H_t = \sqrt{\frac{2}{\pi}}  {\frac{1}{\sqrt{T-t}} e^{-\frac{S_t^2}{2(T-t)}}\mathbf 1_{t\leq \sigma}}$$
is the optimal arbitrage strategy, that is, the hedging strategy for $1_{\sigma > T}$ in the $({\mathbb Q},S)$-market as well as the hedging strategy for 1 in the $({\mathbb P},S)$-market.

Let us now compute the dynamics of $S$ under $\mathbb P$. 
By Girsanov's Theorem (see, e.g., Theorem 41 on page 136 of \cite{protter_stochastic_2003}),
$$ W^{\mathbb P}_t = W^{\mathbb Q}_t - \frac{2}{{\mathbb Q}[\sigma > T]\sqrt{2\pi}} \int\limits_0^{\sigma\wedge t} {\frac{1}{M_{s}} e^{-\frac{S_s^2}{2(T-s)}} \frac{1}{\sqrt{T-s}} ds}$$
is a ${\mathbb P}$-Brownian motion. The dynamics of $S$ under ${\mathbb P}$ are therefore given by
\begin{align} \label{eq:example_BM_dynamic_S}
S_t &= 1 +  W^{\mathbb P}_t + \frac{2}{{\mathbb Q}[\sigma > T]\sqrt{2\pi}} \int\limits_0^{\sigma\wedge t} \frac{e^{-\frac{S_s^2}{2(T-s)}}}{M_s\sqrt{T-s}} ds\\
& = 1 +  W^{\mathbb P}_t + \sqrt{\frac{2}{{\pi}}} \int\limits_0^{\sigma\wedge t} {\frac{1} {1 - 2\mathcal{N} \left( -\frac{S_s}{\sqrt{T-s}} \right) } \frac{e^{-\frac{S_s^2}{2(T-s)}}}{\sqrt{T-s}} ds}
\end{align}

Now, let us discuss the fragility and robustness of the arbitrage in this example in the sense of \cite{guasoni_fragility_2012}. The optimal arbitrage constructed using the predictable stopping time $\sigma = \inf\{ t>0: S_t \le 0\}$ is fragile. Indeed, from (\ref{eq:example_BM_dynamic_S}) we can write the dynamics of $X_t : = \log S_t$ as follows
\begin{equation} \label{eq:example_BM_dynamic_logS}
d X_t = {e^{-X_t}}dW^{\mathbb P}_t  + \left[  \frac{2}{{\mathbb Q}[\sigma > T]\sqrt{2\pi}} {e^{-X_t}} \frac{1}{M_{t}} e^{-\frac{S_t^2}{2(T-t)}} \frac{1}{\sqrt{T-t}} - \frac{1}{2}e^{-2X_t} \right] dt  
\end{equation}
Since in (\ref{eq:example_BM_dynamic_logS}), the drift is locally bounded and the volatility is continuous and nonsingular, by Theorem 2 of \cite{guasoni_fragility_2012}, we conclude that the arbitrage is fragile: for any $\varepsilon>0$, one can find a process $\widetilde S_t$ which admits an equivalent martingale measure and satisfies $1-\varepsilon \leq \frac{\widetilde S(t)}{S(t)}\leq 1+\varepsilon$ a.s. for all $t\in[0,T]$. 

However, we can slightly modify the stopping time $\sigma$ to construct an arbitrage which is not destroyed by small perturbations of the price process as above. More precisely, we choose the predictable stopping time $\sigma$ as the first time when $S_t$ hits a  line with positive slope, that is 
$$\sigma = \inf\{t \ge 0: S_t \le \alpha t\}$$
with $\alpha > 0$. By Proposition 3.2.1.1 in \cite{jeanblanc_mathematical_2009}, 
\begin{align*}
{\mathbb Q}[\sigma > T] &= {\mathbb Q} \left[  \inf_{0 \le t \le T} (W^{\mathbb Q}_t - \alpha t) > -1 \right] \\
& = \mathcal{N} \left( \frac{1 - \alpha T }{\sqrt{T}} \right) - e^{2 \alpha} \mathcal{N} \left( \frac{-1 - \alpha T }{\sqrt{T}} \right)\in(0,1),
\end{align*}
and we can define a measure $\mathbb P$ admitting optimal arbitrage via \eqref{defprobatau}. 
It is easy to see that $S_T > \alpha T$, ${\mathbb P}$-a.s. This allows to construct a simple buy-and-hold arbitrage strategy.
\begin{itemize}
\item If $\alpha T > 1$, buy one unit of $S$ in the begining and hold it until $T$. This strategy yields a profit of $\alpha T - 1$ with probability $1$. 
\item If $\alpha T \leq 1$, introduce the stopping time $\sigma_1 = \inf\{t>0: S_t = \frac{\alpha T}{2}\}$. If $\sigma_1 \leq \frac{T}{2}$, buy one unit of $S$ at $\sigma_1$ and hold it until $T$. Otherwise, do nothing. It is easy to see that $\mathbb P[\sigma_1\leq T/2] = \mathbb Q[\sigma_1\leq T/2, \sigma > T]>0$, which means that this strategy yields a profit greater than $\frac{\alpha T}{2}$ with positive probability. 
\end{itemize}
{This strategy is a simple obvious arbitrage in the sense of \cite{bender_simple_2012}, which means that it is robust and not fragile (see discussion at the end of section \ref{section:setting}).
}

We are going to compute the dynamics of $\log S$ in this case and compare to the results of \cite{guasoni_fragility_2012}.
By Markov property and the law of infimum of Brownian motion with drift, we compute the conditional probability  
\begin{align*}
{\mathbb E}^{\mathbb Q}[1_{\sigma > T}|\mathcal F_t]  =  \left\{ {\begin{array}{ll}
   {0} \qquad \text{if  } {\sigma \leq t}   \\
   {\mathcal{N} \left( \frac{S_t - \alpha T}{\sqrt{T-t}} \right)- e^{2 \alpha (S_t - \alpha t) } \mathcal{N} \left( \frac{-S_t + 2 \alpha t  - \alpha T}{\sqrt{T-t}} \right)} \qquad \text{if } \sigma > t.\\
   \end{array}} \right.
\end{align*}
Denoting
$$ Y^1_t = \frac{S_t - \alpha T}{\sqrt{T-t}}, \qquad Y^2_t = \frac{-S_t + 2 \alpha t  - \alpha T}{\sqrt{T-t}}  $$
and applying the Itô formula, we obtain the dynamics of the conditional law on $\sigma<t$:
\begin{align*}
d{\mathbb E}^{\mathbb Q}[1_{\sigma > T}|\mathcal F_t] &= \left[ \frac{1}{\sqrt{2 \pi}}  \frac{e^{-\frac{(Y_t^1)^2}{2}}}{\sqrt{T-t}} +  \frac{e^{2\alpha (S_t - \alpha t)}}{\sqrt{2 \pi}}  \frac{e^{-\frac{(Y_t^2)^2}{2}}}{\sqrt{T-t}} - \mathcal{N}(Y^2_t) 2 \alpha e^{2\alpha (S_t - \alpha t)}  \right]  dW^{\mathbb Q}_t,\\
\end{align*}
and the dynamics of $M_t$:
$$ 
dM_t = \frac{1}{{\mathbb Q}[\sigma > T]} \left[ \frac{1}{\sqrt{2 \pi}}  \frac{e^{-\frac{(Y_t^1)^2}{2}}}{\sqrt{T-t}} + \frac{ e^{2\alpha (S_t - \alpha t)}}{\sqrt{2 \pi}}  \frac{e^{-\frac{(Y_t^2)^2}{2}}}{\sqrt{T-t}} - \mathcal{N}(Y^2_t) 2 \alpha e^{2\alpha (S_t - \alpha t)}  \right]  dW^{\mathbb Q}_t.
$$
By Girsanov's Theorem,
$$ dW^{\mathbb P}_t = dW^{\mathbb Q}_t - \frac{1}{M_t{\mathbb Q}[\sigma > T]} \left[ \frac{1}{\sqrt{2 \pi}}  \frac{e^{-\frac{(Y_t^1)^2}{2}}}{\sqrt{T-t}} + \frac{e^{2\alpha (S_t - \alpha t)}}{\sqrt{2 \pi}}  \frac{e^{-\frac{(Y_t^2)^2}{2}}}{\sqrt{T-t}} - \mathcal{N}(Y^2_t) 2 \alpha e^{2\alpha (S_t - \alpha t)}  \right]dt$$
is a ${\mathbb P}$-Brownian motion. Finally, the dynamic of $S$ under ${\mathbb P}$ is 
$$ dS_t = dW^{\mathbb P}_t + \frac{1}{M_t{\mathbb Q}[\sigma > T]} \left[ \frac{1}{\sqrt{2 \pi}}  \frac{e^{-\frac{(Y_t^1)^2}{2}}}{\sqrt{T-t}} +  \frac{e^{2\alpha (S_t - \alpha t)}}{\sqrt{2 \pi}}  \frac{e^{-\frac{(Y_t^2)^2}{2}}}{\sqrt{T-t}} - \mathcal{N}(Y_t^2) 2 \alpha e^{2\alpha (S_t - \alpha t)}  \right]dt.$$
Applying Itô's formula once again, we see that $X_t = \log S_t$ satisfies
\begin{align}
\label{eq:example_BM_dynamic_logS_drift}
dX_t & = e^{-X_t} dS_t -\frac{1}{2}e^{-2X_t}dt \nonumber\\
& =  \frac{e^{-X_t}}{M_t{\mathbb Q}[\sigma > T]} \left[ \frac{1}{\sqrt{2 \pi}}  \frac{e^{-\frac{(Y_t^1)^2}{2}}}{\sqrt{T-t}} + \frac{e^{2\alpha (S_t - \alpha t)}}{\sqrt{2 \pi}}  \frac{e^{-\frac{(Y_t^2)^2}{2}}}{\sqrt{T-t}} - \mathcal{N}(Y^2_t) 2 \alpha e^{2\alpha (S_t - \alpha t)}  \right]dt \nonumber\\
&-\frac{1}{2}e^{-2X_t}dt + {e^{-X_t}}dW^{\mathbb P}_t.
\end{align}
The drift in (\ref{eq:example_BM_dynamic_logS_drift}) can be written as a function of $(t, X_t)$, and it is not locally bounded, for example, $1/M$ is unbounded in a neighborhood of $(t, \log (\alpha t))$. So the result of \cite{guasoni_fragility_2012} breaks down. 


\subsection{A robust arbitrage based on the Poisson process}
Another way to ensure robustness of arbitrage with respect to small perturbations is to introduce jumps into the price process dynamics. 
Let $N$ be a standard Poisson process under ${\mathbb Q}$ and assume that $\mathbb F = \mathbb F^N$, which is a quasi left-continuous filtration. Then $S_t = 1 + N_t - t$ is a ${\mathbb Q}$-martingale. We define a predictable stopping time
$\tau = \inf \{ t>0: S_t \le 0\}$ and a new probability measure ${\mathbb P} \ll {\mathbb Q}$ via $d{\mathbb P}|_{\mathcal F_t} = S_{t \wedge \tau} d{\mathbb Q}|_{\mathcal F_t}$. 
The $({\mathbb P},S)$-market admits optimal arbitrage provided $T>1$, because in this case $SP^{\mathbb Q}(1_{S_T > 0}) = {\mathbb Q}[S_T > 0] < 1$. 

{Here, unlike the first example of this section or the Bessel process example discussed in  \cite{guasoni_fragility_2012}, we can prove that the arbitrage is not fragile. Indeed, we fix a real number $\varepsilon > 0$ and construct a simple buy-and-hold
arbitrage strategy as follows:
\begin{itemize}
\item if $S$ jumps on $[0, \varepsilon]$ then we do nothing.
\item if $S$ does not jump on $[0, \varepsilon]$, we buy one unit of $S$ at $\varepsilon$ and hold it until the first jump time of $S$.
\end{itemize} 
Assuming that $T>1$, the process $N$ must jump before $T$, because $S_t> 0, \mathbb{P}$-a.s. This means that this strategy generates a profit greater than $\varepsilon$ with positive probability. Therefore, this is a simple obvious arbitrage in the sense of \cite{bender_simple_2012}, and so it is not fragile (see discussion at the end of section \ref{section:setting}).
}

\subsection{Extension to incomplete markets}

Assume that $S$ is a nonnegative ${\mathbb Q}$ local martingale with only positive jumps starting at 1. Suppose that the conditions in Corollary \ref{cor:optimal_arbitrage_predictable} are fulfilled.  The requirement of quasi-left continuity is not so restrictive, for example, the natural completed filtration of a L\'evy process is 
a quasi-left continuous filtration (see Exercise 8 page 148 of \cite{protter_stochastic_2003}).

Let $a$ be a positive number such that  $aT/2>1$ and define a predictable stopping time by $\sigma = \inf\{t>0: S_t \leq at\}$. 
The measure $\mathbb{P}$ is defined by
\begin{align*}\left. {\frac{d{\mathbb P}}{d{\mathbb Q}}} \right|_{\mathcal F_t}  =  \frac{{\mathbb Q}\left[ \sigma > T | \mathcal F_t \right] }{{\mathbb Q}\left[ \sigma > T \right] }.
\end{align*}
From the economics point of view, this arbitrage represents a bet that the asset price will remain above the line $\alpha t$. Then for any equivalent local martingale measure $\bar {\mathbb Q}$ such that $S$ is a nonnegative $\bar {\mathbb Q}$ local martingale, we have
\begin{align*}
\bar {\mathbb Q}[\sigma \leq T] &\ge \bar {\mathbb Q}[S_{T/2} \le aT/2] = 1 - \bar {\mathbb Q}[S_{T/2} > aT/2] \\
&\ge 1 - \frac{{\mathbb E}^{\bar {\mathbb Q}} [S_{T/2}]}{aT/2} \ge 1 - \frac{1}{aT/2}.
\end{align*}
The superhedging price is $SP^{\mathbb Q}(1_{\sigma > T} ) = \sup_{\bar {\mathbb Q}} \bar {\mathbb Q}[\sigma > T] \le \frac{1}{aT/2} < 1$. Therefore, the $(\mathbb P,S)$-market admits optimal arbitrage. Since $S_T > aT > 2, {\mathbb P}-a.s.$, the presence of arbitrage is robust with respect to small perturbations of $S$. 

\subsection{A variation: building an arbitrage from a bubble}
Let $S$ be a nonnegative ${\mathbb Q}$ local martingale with no positive jumps, satisfying $S_0$ = 1 and $S_t \le \varepsilon < 1$, ${\mathbb Q}$-a.s for $t \ge T$. In other words, the market admits a bubble.

We define
$\sigma = \inf\{t \ge 0: S_t > K\}$ for $K > 1$. In this example, a trader believes that the price of $S$ may not exceed an upper bound $K$. As in previous examples, this trader may construct an arbitrage model $\mathbb{P}$ as in (\ref{defprobatau}), provided that the conditions in Corollary \ref{cor:optimal_arbitrage_predictable} are fulfilled.
  Under any ELMM $\bar {\mathbb Q}$, $S_{\sigma \wedge t}$ is a bounded $\bar {\mathbb Q}$ local martingale and hence a $\bar {\mathbb Q}$ martingale. So we get 
$$1 = {\mathbb E}^{\bar {\mathbb Q}}[S_{\sigma \wedge T}] = K \bar {\mathbb Q}[\sigma \leq T] + S_T \bar {\mathbb Q}[\sigma > T],$$
and therefore
$$\bar {\mathbb Q}[\sigma \leq T] = \frac{1 - S_T\bar {\mathbb Q}[\sigma > T]}{K} > \frac{1-\varepsilon}{K}.$$
The superhedging price of $1_{\sigma > T}$ is
$$ \sup_{\bar {\mathbb Q} \in ELMM({\mathbb Q},S)} {\mathbb E}^{\bar {\mathbb Q}}[1_{\sigma > T}] < 1 - \frac{1-\varepsilon }{K}.$$

\subsection{A joint bet on an asset and its volatility} \label{example:stochastic_volatility}
Let 
$$
S_t = 1 +  \int_0^t \sigma^S_udW^{\mathbb Q}_u,
$$
where $W^{\mathbb Q}$ is a Brownian motion and $\sigma^S$ is a volatility process
assumed to be continuous. Let $\underline\sigma< \sigma^S_0$
and define stopping times as follows:
$$
\sigma_1 = \inf\{t>0: S_t \leq 0\},\qquad \sigma_2 = \inf\{t>0: \sigma^S_t
\leq \underline\sigma\},\qquad \sigma = \sigma_1\wedge \sigma_2.
$$
Assume that ${\mathbb Q}[\sigma_1 \leq T, \sigma_2 > T] > 0$ and the conditions in Corollary \ref{cor:optimal_arbitrage_predictable} are fulfilled. This choice of the stopping time represents a bet that $S$ will not hit $0$ and its volatility will stay above $\underline{\sigma}$ up to time $T$. Then, under any equivalent local martingale measure $\bar {\mathbb Q}$, we have 
\begin{align*}
\bar {\mathbb Q}[\sigma \leq T] &= \bar {\mathbb Q}[\sigma_1 \leq T, \sigma_2 > T] + \bar {\mathbb Q} [\sigma_2 \leq T] \\ &\geq \bar {\mathbb Q}[B^*_{\underline \sigma^2 T}\leq -1, \sigma_2 > T] + \bar {\mathbb Q}[\sigma_2 \leq T] \geq \bar {\mathbb Q}[B^*_{\underline
  \sigma^2 T}\leq -1]>0,
\end{align*}
where $B^*_t = \inf_{0\leq u \leq t} B_u$ and $B$ is the Brownian
motion such that 
$$\int\limits_0^t {\sigma^S_udW^{\mathbb Q}_u}= B_{\int_0^t (\sigma^S_u)^2 du}.$$ 
Because $\bar {\mathbb Q}[\sigma \leq T]$ is bounded from below by the quantity $\bar {\mathbb Q}[B^*_{\underline\sigma^2 T}\leq -1]$. This quantity is positive and does not depend on $\bar {\mathbb Q}$, since $B$ is a Brownian motion under $\bar {\mathbb Q}$. Therefore $\sup_{\bar {\mathbb Q}} \bar Q[\sigma > T]$ is bounded away from one.\\
The superhedging price $SP^{\mathbb Q}(1_{\sigma > T} ) = \sup_{\bar {\mathbb Q}} \bar {\mathbb Q}[\sigma > T] < 1$.

\subsection{A variation: betting on the square bracket}
The construction in Example \ref{example:stochastic_volatility} can be modified as follows. Let $S$ be a nonnegative continuous semimartingale starting at 1 and be a ${\mathbb Q}$ local martingale. We define 
$$
\sigma_1 = \inf\{t>0: S_t \le 0\},\quad \sigma_2 = \inf\{t>0:[S]_t \le -a + bt\},\quad \sigma = \sigma_1\wedge \sigma_2,
$$
where $a, b$ are positive constants. Suppose that ${\mathbb Q}[\sigma_1 \leq T, \sigma_2 > T] > 0$ and the conditions in Corollary \ref{cor:optimal_arbitrage_predictable} are fulfilled. Then under any ELMM $\bar {\mathbb Q}$, we compute
\begin{align*}
\bar {\mathbb Q}[\sigma \leq T] &= \bar {\mathbb Q}[\sigma_1 \leq T, \sigma_2 > T] + \bar {\mathbb Q} [\sigma_2 \leq T] \\ &\geq \bar {\mathbb Q}[B^*_{[S]_T}\leq -1, \sigma_2 > T] + \bar {\mathbb Q}[\sigma_2 \leq T] \geq \bar {\mathbb Q}[B^*_{bT-a}\leq -1]>0,
\end{align*}
This implies that the superhedging price
$SP^{\mathbb Q}(1_{\sigma > T} ) = \sup_{\bar {\mathbb Q}} \bar {\mathbb Q}[\sigma > T] < 1$.



\end{document}